\documentclass[a4paper,10pt, onecolumn,conference]{ieeeconf}  

\IEEEoverridecommandlockouts                              

\usepackage{graphicx} 
\usepackage{mathptmx} 
\usepackage{amsmath} 
\usepackage{amssymb}  
\usepackage{datetime}
\usepackage{hyperref}
\usepackage{color}

\newtheorem{lem}{Lemma}
\newtheorem{thm}{Theorem}

\def\rank{\mathrm{rank}}
\DeclareMathAlphabet{\bit}{OML}{cmm}{b}{it}

\def\<{\leqslant}           
\def\>{\geqslant}           

\def\d{\partial}
\def\wh{\widehat}
\def\wt{\widetilde}

\def\Re{\mathrm{Re} }   
\def\Im{\mathrm{Im} }   

\def\mR{{\mathbb R}}
\def\mC{\mathbb{C}}

\def\Tr{\mathrm{Tr}}
\def\im{\mathrm{im}}

\def\rT{\mathrm{T}}

\def\bS{\mathbf{S}}

\def\bP{\mathbf{P}}
\def\bE{\mathbf{E}}

\def\[[[{[\![\![}
\def\]]]{]\!]\!]}

\def\bra{\langle}
\def\ket{\rangle}

\def\re{\mathrm{e}}
\def\rd{\mathrm{d}}

\def\bJ{\mathbf{J}}

\def\x{\times}
\def\ox{\otimes}

\def\fF{\mathfrak{F}}

\def\fH{\mathfrak{H}}

\def\cD{\mathcal{D}}

\def\cG{\mathcal{G}}
\def\cI{\mathcal{I}}

\def\cov{\mathbf{cov}}

\def\eps{\epsilon}
\def\ups{\upsilon}
\def\Ups{\Upsilon}

\begin{document}
\title{\LARGE \bf 
Optimization of Partially Isolated Quantum Harmonic Oscillator Memory Systems by Mean Square Decoherence Time Criteria${}^*$}

\author{Igor G. Vladimirov$^{1}$, \quad Ian R. Petersen$^{2}$%
\thanks{${}^*$This work is supported by the Australian Research Council grant DP240101494.}
\thanks{$^{1,2}$School of Engineering, Australian National University, ACT 2601, Canberra,
Australia:
        {\tt\small igor.g.vladimirov@gmail.com,
        i.r.petersen@gmail.com}.}%
}

\maketitle
\thispagestyle{empty}
\pagestyle{plain}

\begin{abstract}
This paper is concerned with open quantum harmonic oscillators with position-momentum system variables, whose internal dynamics and interaction with the environment are governed by linear quantum stochastic differential equations. A recently proposed  approach to such systems as Heisenberg picture quantum memories  exploits their ability to approximately retain initial conditions over a decoherence horizon. Using the
quantum  memory decoherence time defined  previously in terms of a fidelity threshold on a weighted mean-square deviation of the system variables from their initial values, we apply this approach to a partially isolated subsystem of the oscillator, which is not directly affected by the external fields. The partial isolation leads to an appropriate system decomposition and a qualitatively different short-horizon asymptotic behaviour of the deviation, which yields a longer decoherence time in the high-fidelity limit.
The resulting approximate decoherence time  maximization over the energy parameters for improving the quantum memory performance is discussed for a coherent feedback interconnection of such systems.
\end{abstract}

\section{INTRODUCTION}

The quantum communication and quantum information processing technologies,  which undergo extensive theoretical and practical development, substantially rely on the possibility to manipulate and engineer quantum mechanical  systems with novel properties exploiting nonclassical resources. The latter come from the noncommutative operator-valued structure of quantum variables and quantum states reflected in the Heisenberg and Schr\"{o}dinger pictures of quantum dynamics and quantum probability, and the theory of quantum measurement   \cite{H_2001}.
One of the aspects of these approaches is that the fragile nature of quantum dynamics and quantum states,  involving the  atomic and subatomic scales \cite{LL_1981},  complicates the state initialization for such systems and their isolation from the environment in a controlled fashion. In particular, the controlled isolation is important in the context of quantum computation paradigms \cite{NC_2000}  where unitary transformations,  performed with closed quantum systems \cite{S_1994} between the measurements,   play a crucial role.

The absence of dissipation  makes the perfectly   reversible unitary dynamics of an isolated quantum  system (including the preservation of the Hamiltonian and of the system variables or states in the zero-Hamiltonian case)  particularly useful for storing quantum information. The storage phase is preceded by the ``writing'' stage and followed by the ``reading'' stage  which can employ, for example, photonics-based  platforms using light-matter interaction effects  (see \cite{FCHJ_2016,HEHBANS_2016,WM_2008,WLZ_2019,YJ_2014} and references therein).
However,
the ability to retain quantum states or dynamic variables over the course of time, which is required for the storage phase, is corrupted by the unavoidable coupling of the quantum system to its environment such as external fields.

The system-field interaction, even when it is weak,  gives rise to quantum noise which makes the system drift away in a dissipative fashion  from its initial condition as opposed to the ideal case of isolated zero-Hamiltonian system dynamics.
The resulting open quantum systems are modelled  by Hudson-Parthasarathy quantum stochastic differential equations (QSDEs) \cite{HP_1984,P_1992}  which are linear for open quantum harmonic oscillators (OQHOs) \cite{JNP_2008,NY_2017,P_2017,ZD_2022} and quasi-linear for finite-level quantum systems \cite{EMPUJ_2016,VP_2022_SIAM}. While the latter are particularly relevant for modelling qubit registers in the quantum computing context, continuous variables systems (such as the OQHOs  with quantum positions and momenta) are also employed in quantum information theory including the Gaussian quantum channel models \cite{H_2019}.

Using the quantum calculus framework,  an approach has recently been proposed in \cite{VP_2024_ANZCC,VP_2024_EJC} for both classes of open quantum systems to their performance as a quantum memory in the storage phase. The system performance index is a decoherence time horizon at which a weighted mean-square deviation of the system variables from their initial values reaches a given  fidelity threshold (note that different yet related decoherence measures are also discussed in \cite{VP_2023_SCL} and references therein). The memory decoherence time (or its high-fidelity asymptotic approximation) can therefore be maximized over the energy and coupling parameters of a quantum network in order to improve its ability to approximately retain  the initial conditions.

In  the present paper, we follow the approach of \cite{VP_2024_ANZCC} to OQHOs as Heisenberg picture  quantum memory systems with position-momentum variables, whose internal dynamics and interaction with the environment are governed by linear QSDEs.
Using the
previously defined memory decoherence time,
we apply this approach to a partially isolated subsystem of the oscillator, which is affected by the external fields only indirectly through another subsystem. The  partial subsystem isolation from the external fields and the related system decomposition hold under a certain rank condition  on the system-field  coupling matrix and form an important special case which was not considered in \cite{VP_2024_ANZCC}. This setting leads to a qualitatively different short-horizon asymptotic behaviour of the mean-square deviation of the subsystem variables from their initial values and yields a longer decoherence time in the high-fidelity limit.
The maximization of the resulting approximate decoherence time  over the energy parameters for improving the quantum memory performance is discussed for a coherent feedback interconnection of such systems involving the direct energy coupling and field-mediated coupling \cite{ZJ_2011a} between the constituent OQHOs.

The paper is organised as follows.
Section~\ref{sec:sys} provides a minimum background material on QSDEs for open quantum stochastic systems.
Section~\ref{sec:osc} specifies the class of OQHOs with position-momentum system variables and linear-quadratic energetics.
Section~\ref{sec:dev} considers weighted deviations of the system variables from their initial conditions and discusses a partially    isolated subsystem along with the system decomposition.
Section~\ref{sec:meansquare} reviews the mean-square quantification of such deviations and the memory decoherence time and studies the short-horizon  and high-fidelity asymptotic behaviour of these functionals for the partially isolated subsystem.
Section~\ref{sec:two} solves the approximate decoherence time maximization problem for a partially isolated subsystem of a coherent feedback interconnection of two OQHOs with direct energy and indirect field-mediated coupling.

\section{QUANTUM STOCHASTIC DYNAMICS}
\label{sec:sys}

As mentioned above, in comparison with the idealisation of completely isolated quantum dynamics, more realistic settings consider open quantum systems which interact with the environment, such as quantum fields
or classical measuring devices.
In the framework of the Hudson-Parthasarathy  quantum stochastic calculus \cite{HP_1984,P_1992}, an open quantum system subject to external quantum fields, as shown in Fig.~\ref{fig:sys},
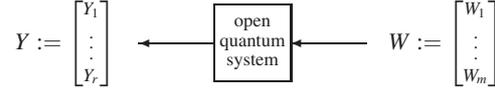
\begin{figure}[htbp]
{\centering
\unitlength=1mm
\linethickness{0.4pt}
\begin{picture}(40.00,7)
    \put(20,3.5){\makebox(0,0)[cc]{\scriptsize open}}
    \put(15,-4){\framebox(10,10)[cc]{\scriptsize quantum}}
    \put(20,-1.5){\makebox(0,0)[cc]{\scriptsize system}}
    \put(35,1){\vector(-1,0){10}}
    \put(15,1){\vector(-1,0){10}}
    \put(38,1){\makebox(0,0)[lc]{\small $W:={\scriptsize \begin{bmatrix}
      W_1\\
      \vdots\\
      W_m
    \end{bmatrix} }
    $}}
    \put(2,1){\makebox(0,0)[rc]{\small $Y:={\scriptsize \begin{bmatrix}
      Y_1\\
      \vdots\\
      Y_r
    \end{bmatrix} }$}}
\end{picture}\vskip3mm}
\caption{An open quantum system with vectors $W$, $Y$ of input quantum Wiener processes and output fields.}
\label{fig:sys}
\end{figure}
is governed by QSDEs
\begin{equation}
\label{dXdYgen}
    \rd X  = \cG(X) \rd t - i[X, L^\rT] \rd W,
    \qquad
    \rd Y  = 2DJ L \rd t + D\rd W
  \end{equation}\vskip-1.5mm\noindent
  (the time arguments will often be omitted for brevity),
  which describe the Heisenberg  evolution \cite{S_1994} of the internal and output variables.  Here, $i:=\sqrt{-1}$ is the imaginary unit, and the reduced Planck constant is set to  $\hslash = 1$ using atomic units for convenience. Also, $X:=(X_k)_{1\< k \< n}$, $Y$, $W$  are vectors of internal, output and input dynamic variables, respectively, and $L := (L_k)_{1\< k\< m}$ is a vector of system-field coupling operators.   Unless indicated otherwise,  vectors are organised as columns.  Accordingly,   $[X,L^\rT]:= (X_jL_k - L_k X_j)_{1\< j \< n, 1\< k \< m} = XL^\rT - (LX^\rT)^\rT$ is the commutator matrix.
The system variables  $X_1, \ldots, X_n$ and the  output field variables $Y_1, \ldots, Y_r$ are time-varying self-adjoint operators on the system-field tensor-product  Hilbert space
\begin{equation}
\label{fH}
    \fH:= \fH_0\ox \fF,
\end{equation}
with the system variables acting initially, at time $t=0$,  on a Hilbert space $\fH_0$. The input bosonic field variables $W_1, \ldots, W_m$  are quantum Wiener processes which are time-varying self-adjoint operators on a symmetric Fock space $\fF$.
The map $\cG$ in the drift term of the first QSDE in (\ref{dXdYgen}) is the Gorini-Kossakowski-Sudar\-shan-Lindblad superoperator \cite{GKS_1976,L_1976}  acting
on a system operator $\xi$ (such as a function of the system variables $X_1, \ldots, X_n$) as
\begin{equation}
\label{cG}
\cG(\xi)
   := i[H,\xi]
     +
     \cD(\xi).
\end{equation}
This is a quantum counterpart of the infinitesimal generators of classical diffusion processes \cite{KS_1991}, with $\cD$ the decoherence superoperator
given by
\begin{equation}
\label{cD}
    \cD(\xi)
    :=
     \frac{1}{2}
    ([L^\rT,\xi]\Omega L + L^\rT \Omega [\xi,L]).
\end{equation}
The Hamiltonian $H$ in (\ref{cG}) and the system-field coupling operators $L_1, \ldots, L_m$ in (\ref{dXdYgen}), (\ref{cD})  are self-adjoint operator-valued functions of $X_1, \ldots, X_n$. The matrix $D \in \mR^{r\x m}$ in (\ref{dXdYgen}) consists of $r\< m$  conjugate  rows of a permutation matrix of order $m$, where both  $r$ and $m$ are even. Thus,  $D$ is a co-isometry: $DD^\rT = I_r$, where $I_r$ is the identity matrix of order $r$. The meaning of $D$ is that it selects some of the $m$ output fields. For example, the case of $r=m$ and $D=I_m$ corresponds to selecting all of the output fields.
The vector $W$ of quantum Wiener processes $W_1, \ldots, W_m$, which drive the QSDEs (\ref{dXdYgen}), has the quantum Ito  table
\begin{equation}
\label{Omega}
    \rd W \rd W^\rT
    =
    \Omega \rd t,
    \qquad
    \Omega
    :=
    I_m + iJ
    =
    \Omega^*\succcurlyeq 0,
\end{equation}
where $(\cdot)^*:= {\overline{(\cdot)}}{}^{\rT}$ is the complex conjugate transpose.
Here, $J := \Im \Omega$ is a real antisymmetric matrix of order $m$ given by
\begin{equation}
\label{JJ}
    J
    :=
    I_{m/2}\ox\bJ,
    \qquad
            \bJ
        : =
        {\begin{bmatrix}
        0 & 1\\
        -1 & 0
    \end{bmatrix}}
\end{equation}
and specifying the two-point canonical commutation relations (CCRs)
\begin{equation}
\label{WWcomm}
  [W(s), W(t)^\rT]
  =
  2i \min(s,t)J,
  \qquad
   s, t \> 0,
\end{equation}
where $J$ represents its tensor product $J\ox \cI_\fF$ with the identity operator $\cI_\fF$ on the Fock space $\fF$.

\section{OPEN QUANTUM HARMONIC OSCILLATORS}
\label{sec:osc}

The CCRs (\ref{WWcomm}) are typical for unbounded operators of position-momentum type \cite{S_1994} on a dense domain of an infinite-dimensional Hilbert space. Such quantum variables, when they are considered as the system variables $X_1, \ldots, X_n$ with an even
\begin{equation}
\label{nnu}
    n: = 2\nu,
\end{equation}
are provided by the quantum mechanical  positions $q_1, \ldots, q_\nu$ and momenta $p_1, \ldots, p_\nu$ which can be implemented  as the multiplication and differential operators
\begin{equation}
\label{Xqp}
    X_{2k-1}
    := q_k,
    \qquad
    X_{2k} := p_k:= -i\d_{q_k},
    \qquad
    k = 1, \ldots, \nu,
\end{equation}
on the Schwartz space of rapidly decreasing functions \cite{RS_1980} on $\mR^\nu$. At any time $t\>0$,   these operators satisfy the one-point CCRs
\begin{equation}
\label{Xqpcomm}
    [X(t),X(t)^\rT]
    =
    2i\Theta,
    \qquad
    \Theta := \frac{1}{2}I_\nu\ox \bJ.
\end{equation}
As continuous quantum variables, the conjugate position-momentum pairs (\ref{Xqp}) are qualitatively different from those in finite-level quantum systems \cite{EMPUJ_2016,VP_2022_SIAM} and   result from quantizing the classical positions and momenta of Hamiltonian mechanics \cite{A_1989}.
They are employed as system variables in OQHOs \cite{JNP_2008,NY_2017,P_2017,ZD_2022}. For an OQHO,  the Hamiltonian $H$ and the system-field coupling operators $L_1, \ldots, L_m$  are quadratic and linear functions of the system variables specified by an energy matrix $R = R^\rT \in \mR^{n \x n}$ and a coupling matrix $M \in \mR^{m\x n}$ as
\begin{equation}
\label{HLOQHO}
    H
    :=
    \frac{1}{2} X^\rT R X,
    \qquad
    L := MX.
\end{equation}
The general QSDEs (\ref{dXdYgen}) then acquire the form of linear QSDEs
\begin{equation}
\label{dXdY_OQHO}
    \rd X  = AX \rd t +B\rd W,
    \qquad
    \rd Y  = CX \rd t + D\rd W,
\end{equation}
where the matrices  $
    A
    \in \mR^{n\x n}$,
$
    B
    \in \mR^{n\x m}$,
$
    C
    \in \mR^{r\x n}
$ are computed as
\begin{equation}
\label{ABC}
    A := A_0 + \wt{A},
    \qquad
    B := 2\Theta M^\rT,
    \qquad
    C:= 2DJM,
\end{equation}
with
\begin{equation}
\label{AAOQHO}
        A_0 := 2\Theta R,
    \qquad
    \wt{A}
    := 2\Theta M^\rT J M,
\end{equation}
in terms of the CCR matrices $J$, $\Theta$ from (\ref{JJ}),  (\ref{Xqpcomm}) and the energy and coupling matrices $R$, $M$ from (\ref{HLOQHO}).  In view of (\ref{AAOQHO}),  the dynamics matrix $A$ in (\ref{ABC}) depends linearly on $R$ and quadratically on $M$. Moreover, in addition to the noncommutative operator-valued nature of the quantum variables,
 the specific representation of the coefficients of the QSDEs (\ref{dXdY_OQHO})  in terms of the commutation, energy and coupling parameters imposes quantum physical realizability constraints \cite{JNP_2008} in contrast to classical SDEs \cite{KS_1991}.

 Despite the qualitative differences, OQHOs share some common features with classical linear stochastic systems due to the linearity of the QSDEs (\ref{dXdY_OQHO}). In particular, the solution of the first QSDE  is given by
\begin{equation}
\label{Xsol}
    X(t)
    =
    \re^{tA} X(0)
    +
    Z(t),
    \qquad
    Z(t)
    :=
    \int_0^t
    \re^{(t-s)A}
    B
    \rd W(s)
\end{equation}
for any $t\> 0$ and consists of the system responses to the initial condition $X(0)$ and the driving quantum Wiener process $W$ over the time interval $[0,t]$, with $Z(0)=0$. Here, the quantum process
$Z$ of $n$ time-varying self-adjoint operators $Z_1, \ldots, Z_n$  is adapted to the filtration of the Fock space $\fF$, and hence, $X(0)$ and $Z(t)$ commute with each other as operators acting on different spaces ($\fH_0$ and $\fF$):
\begin{equation}
\label{XZcomm}
    [X(0),Z(t)^\rT] = 0,
    \qquad
    t \> 0 .
\end{equation}
Furthermore, if the system-field quantum state on the space (\ref{fH}) has the form
\begin{equation}
\label{rho}
    \rho:= \rho_0\ox \ups,
\end{equation}
where $\rho_0$ is the initial system state on $\fH_0$, and $\ups$ is the vacuum field state \cite{P_1992} on $\fF$, then $Z$ is a Gaussian quantum process (see, for example, \cite[Section~3]{VPJ_2018a}), which is statistically independent  of $X(0)$ and has zero mean:
\begin{equation}
\label{EZ}
    \bE Z(t) = (\Tr (\ups Z_k(t)))_{1\< k \< n} = 0.
\end{equation}
Here,  use is made of the quantum expectation $
    \bE \zeta
    :=
    \Tr (\rho\zeta)
$
 of an operator $\zeta$ on the system-field space (\ref{fH}) over the quantum state   (\ref{rho}),
which reduces in (\ref{EZ})  to the averaging over the vacuum state $\ups$ since the operators $Z_k$ act  on the Fock space $\fF$. In combination with the independence between $X(0)$ and $Z(t)$, (\ref{EZ}) implies that
\begin{equation}
\label{EXZ}
  \bE(X(0)Z(t)^\rT) = \bE X(0) \bE Z(t)^\rT = 0,
\end{equation}
where
$
    \bE X(0) = (\Tr(\rho_0 X_k(0)))_{1\< k \< n}
$.
The above  properties of $Z$ hold regardless of a particular structure of the initial system state $\rho_0$.

\section{PARTIAL SUBSYSTEM ISOLATION}
\label{sec:dev}

The process $Z$ in (\ref{Xsol}), caused by the interaction of the OQHO with the external fields,  plays the role of a noise which corrupts the ability of the system as a quantum memory to retain the initial condition $X(0)$. Similarly to \cite{VP_2024_ANZCC,VP_2024_EJC},  we describe the quantum memory performance in terms of a quadratic form
\begin{equation}
\label{Q}
    Q(t) := \xi(t)^\rT \Sigma\xi(t)
\end{equation}
of the deviation
\begin{equation}
\label{xi}
     \xi(t)
    :=
    X(t)-X(0)
    =
    \alpha_t X(0) + Z(t)
\end{equation}
of the system variables at  time $t\> 0$ from their initial values, where
\begin{equation}
\label{exp}
  \alpha_t := \re^{tA}-I_n,
\end{equation}
with $\xi(0) = 0$, $Q(0) = 0$ and $\alpha_0=0$.
Here,
$\Sigma$ is a real positive semi-definite symmetric matrix of order $n$ factorised as
\begin{equation}
\label{FF}
    \Sigma := F^\rT F,
    \qquad
    F \in \mR^{s\x n},
    \qquad
    s
    :=
    \rank \Sigma \< n,
\end{equation}
so that the rows of $F$ specify the coefficients  of $s$ independent linear combinations of the deviations in (\ref{xi}) forming the vector
\begin{equation}
\label{Fxi}
    \eta(t):= F\xi(t) = FX(t) - FX(0) = F\alpha_t X(0) + FZ(t),
\end{equation}
in terms of which the quantum process $Q$ in (\ref{Q}) is represented as
\begin{equation}
\label{Qeta}
  Q(t)
  :=
  \eta(t)^\rT \eta(t).
\end{equation}
As discussed below,  the matrix $F$  can be used for selecting  those linear combinations (in particular, a subset) of the system variables,  whose dynamics are affected by the external fields only indirectly and thus pertain to a partially isolated subsystem of the OQHO.

\begin{lem}
\label{lem:iso}
Suppose the coupling matrix $M$ of the OQHO in (\ref{HLOQHO}) satisfies
\begin{equation}
\label{Mrank}
    d:= n-\rank M>0.
\end{equation}
Then for any $s\< d$,
there exists a full row rank matrix $F \in \mR^{s\x n}$ such that the vector
\begin{equation}
\label{FX}
    \varphi(t):= FX(t)
\end{equation}
of $s$ time-varying self-adjoint operators on the space (\ref{fH})
satisfies the ODE
\begin{equation}
\label{phidot}
    \mathop{\varphi}^{\centerdot}(t) = G X(t),
    \qquad
    G := FA_0,
\end{equation}
where $A_0$ is the matrix from (\ref{AAOQHO}). \hfill$\square$
\end{lem}
\begin{proof}
Under the condition (\ref{Mrank}), the rows of the matrix $M^\rT \in \mR^{n\x m}$ are linearly  dependent and its columns span in $\mR^n$  a proper subspace $\im(M^\rT)$ of dimension  $\rank M$, whose orthogonal complement $\ker M$  has dimension $d$.
In combination with the nonsingularity of the CCR matrix $\Theta$ in (\ref{Xqpcomm}), this implies that for any $s\< d$, there exists a full row rank matrix $F \in \mR^{s \x n}$ such that the rows of $F\Theta$ are orthogonal to all the columns of $M^\rT$, that is,  $F\Theta M^\rT = 0$.  By the parameterization of the matrix $B$ in (\ref{ABC}), any such $F$ satisfies
\begin{equation}
\label{FB0}
    FB =  2F \Theta M^\rT = 0,
\end{equation}
and hence, the left multiplication of the first QSDE in (\ref{dXdY_OQHO}) by $F$ yields
\begin{equation}
\label{dphi}
    \rd \varphi
    =
    FA X \rd t + FB\rd W
    =
    FA X \rd t
\end{equation}
for the quantum process (\ref{FX}).
Since the diffusion term in (\ref{dphi}) vanishes, this QSDE is an ODE:  $\dot{\varphi}= FAX$,  which establishes (\ref{phidot}) since
\begin{equation}
\label{FAG}
    FA = FA_0 + FB JM = G
\end{equation}
by the first equality in (\ref{ABC}) and the representation $\wt{A} = BJM$ of $\wt{A}$ in (\ref{AAOQHO}) in terms of  $B$.
\end{proof}

Since $\rank M \< \min (m,n)\< m$, then a sufficient condition for (\ref{Mrank})  is provided by the inequality
$    n>m
$.
Furthermore, if the null spaces of the matrices $F$, $G$ in (\ref{phidot}) satisfy $\ker F\subset \ker  G$ and hence,
$
    G = N F
$
for some matrix $N \in \mR^{s \x s}$, then the ODE in (\ref{phidot}) becomes autonomous:
$
    \dot{\varphi}= N \varphi
$, in which case, the dynamics of $\varphi$ are completely isolated from the external fields.

More generally,
the matrix $F$ from Lemma~\ref{lem:iso} can be augmented by a matrix $T \in \mR^{(n-s) \x n}$ to a nonsingular $(n\x n)$-matrix:
\begin{equation}
\label{SFT}
    \det S\ne 0,
    \qquad
    S :=
    {\begin{bmatrix}
      F\\
      T
    \end{bmatrix}}
    \in \mR^{n\x n}.
\end{equation}
By partitioning the inverse matrix $S^{-1}$ into blocks $S_1 \in \mR^{n\x s}$ and $S_2 \in \mR^{n \x (n-s)}$ as
$
    S^{-1}
    :=
    \begin{bmatrix}
      S_1 & S_2
    \end{bmatrix}
$
and introducing a quantum process $\psi$ of $n-s$ time-varying self-adjoint operators on (\ref{fH}),
so that
\begin{equation}
\label{Xzeta}
    X = S^{-1}\zeta
    =
    S_1 \varphi + S_2 \psi,
    \qquad
    \zeta
    :=
    {\begin{bmatrix}
      \varphi\\
      \psi
    \end{bmatrix}},
    \qquad
    \psi := T X,
\end{equation}
it follows that the first QSDE in (\ref{dXdY_OQHO}) can be decomposed into an ODE and a QSDE with respect to the transformed system variables as
\begin{equation}
\label{phipsi}
    \mathop{\varphi}^{\centerdot} = a_{11}\varphi + a_{12}\psi,
    \qquad
    \rd \psi
    =
    (a_{21}\varphi + a_{22}\psi) \rd t + b \rd W,
\end{equation}
where
\begin{equation}
\label{ab}
    a:=
        SAS^{-1}
        :=
    {\begin{bmatrix}
      a_{11} & a_{12}\\
      a_{21} & a_{22}
    \end{bmatrix}}
    =
    {\begin{bmatrix}
      GS_1 & GS_2\\
      TAS_1 & TAS_2
    \end{bmatrix}},
    \qquad
    b := TB.
\end{equation}
Therefore, the internal dynamics of the OQHO can be represented as a feedback interconnection of linear systems $\Phi$ and $\Psi$, with the corresponding vectors $\varphi$, $\psi$  of system variables, where only $\Psi$ interacts directly with the quantum Wiener process $W$, as shown in Fig.~\ref{fig:PhiPsi}.
\begin{figure}[htbp]
{\centering
\unitlength=0.8mm
\linethickness{0.4pt}
\begin{picture}(35.00,15)
    \put(15,-1){\makebox(0,0)[cc]{\small $\varphi$}}
    \put(15,11){\makebox(0,0)[cc]{\small $\psi$}}
    \put(0,0){\framebox(10,10)[cc]{\small$\Phi$}}
    \put(20,0){\framebox(10,10)[cc]{\small$\Psi$}}
    \put(40,5){\vector(-1,0){10}}
    \put(20,8){\vector(-1,0){10}}
    \put(10,2){\vector(1,0){10}}
    \put(43,5){\makebox(0,0)[lc]{\small $W$}}
\end{picture}
\caption{A schematic representation of (\ref{phipsi}) as an interconnection of systems $\Phi$, $\Psi$, where $\Phi$ is affected by the external fields $W$ only through $\Psi$ which interacts with $W$.}
\label{fig:PhiPsi}}
\end{figure}
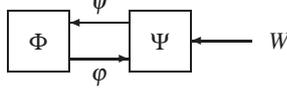
With a slight abuse of notation, the transfer functions  of these systems are computed in terms of the matrices (\ref{ab}) as
\begin{align}
\label{Phi}
    \Phi(u)
    & = (uI_s - a_{11})^{-1}a_{12},\\
\label{Psi}
    \Psi(u)
    & =
    \begin{bmatrix}
      \Psi_1(u) & \Psi_2(u)
    \end{bmatrix}
    =
    (uI_{n-s} - a_{22})^{-1}
    \begin{bmatrix}
      a_{21} & b
    \end{bmatrix}
\end{align}
and relate the Laplace transforms
\begin{align}
\label{phihat_psihat}
    \wh{\varphi}(u)
    & :=
    \int_0^{+\infty}
    \re^{-ut}
    \varphi(t)
    \rd t,
    \qquad
    \wh{\psi}(u)
    :=
    \int_0^{+\infty}
    \re^{-ut}
    \psi(t)
    \rd t,\\
\label{What}
    \wh{W}(u)
    & :=
    \int_0^{+\infty}
    \re^{-ut}
    \rd W(t)
\end{align}
(with the integral in (\ref{What}) being of the Ito type in contrast to those in (\ref{phihat_psihat}))
of the quantum processes $\varphi$, $\psi$, $W$   for all $u \in \mC$ with sufficiently large $\Re u>0$ for convergence of the integrals as
\begin{align}
\label{phihat1}
    \wh{\varphi}(u)
    & =
     \chi_1(u) \zeta(0)
     +
    \Phi(u) \wh{\psi}(u),\\
\nonumber
    \wh{\psi}(u)
    & =
     \chi_2(u) \zeta(0)
     +
    \Psi(u)
    {\begin{bmatrix}
        \wh{\varphi}(u)\\
        \wh{W}(u)
    \end{bmatrix}}\\
\label{psihat1}
    & =
     \chi_2(u) \zeta(0)
     +
    \Psi_1(u)\wh{\varphi}(u) + \Psi_2(u)\wh{W}(u).
\end{align}
Here, $\chi_1$, $\chi_2$ are $(s \x n)$ and  $(n-s)\x n$-blocks of the transfer function
$    \chi(u)
    :=
    {\small\begin{bmatrix}
      \chi_1(u)\\
      \chi_2(u)
    \end{bmatrix}}
    =
    (u I_n - a)^{-1}
$
associated with the response of the quantum process $\zeta$ in (\ref{Xzeta}) to its initial condition $\zeta(0)$. The equations (\ref{phihat1}), (\ref{psihat1}) can be solved for $\wh{\varphi}$ as
\begin{align}
\nonumber
    \wh{\varphi}(u)
    & =
    (\chi_1(u) + \Gamma(u) \chi_2(u))\zeta(0)
    +
    \Gamma(u) \Psi_2(u) \wh{W}(u) \\
\label{phisol}
    & =
    \begin{bmatrix}
    \chi_1(u) + \Gamma(u) \chi_2(u)
    &
    \Gamma(u) \Psi_2(u)
    \end{bmatrix}
    {\begin{bmatrix}
      \zeta(0)\\
      \wh{W}(u)
    \end{bmatrix}},
\end{align}
where $\Gamma$ is an auxiliary transfer function associated with (\ref{Phi}), (\ref{Psi}) by
\begin{equation}
\label{Gamma}
  \Gamma(u)
  :=
  \Phi(u) (I_{n-s} - \Psi_1(u)\Phi(u))^{-1}.
\end{equation}
The Laplace transform of the quantum process $\eta$ in (\ref{Fxi}) is then found by substituting (\ref{phisol}) into
\begin{equation}
\label{etahat}
    \wh{\eta}(u)
    :=
    \int_0^{+\infty}
    \re^{-ut}
    \eta(t)\rd t
    =
    \wh{\varphi}(u) - \frac{1}{u}\varphi(0),
\end{equation}
where $\varphi(0) = FX(0)$ from (\ref{FX})  is the subvector of $\zeta(0)$ in (\ref{Xzeta}). Note that the resulting frequency-domain representation (\ref{etahat}) of $\eta$ does not depend on a particular choice of the matrix $T$ in (\ref{SFT}).
In view of (\ref{phisol}), the subsystem $\Phi$ with the vector $\varphi$ of transformed system variables can be approximately isolated from the external fields $W$  by making the transfer function $\Gamma \Psi_2$ ``small'' in a suitable sense.  By (\ref{Gamma}), this requires ``smallness'' not only of $\Phi$ and $\Psi_2$, but also of $\Psi_1 \Phi$ due to the presence of the feedback loop in Fig.~\ref{fig:PhiPsi}, which is typical for small-gain theorem arguments (see, for example, \cite{DJ_2006} and references therein).

\section{MEAN-SQUARE MEMORY DECOHERENCE TIME}
\label{sec:meansquare}

Following \cite{VP_2024_ANZCC,VP_2024_EJC}, for any given time $t\> 0$, we will quantify the ``size'' of the deviation $\eta(t)$ in  (\ref{Fxi})   by a mean-square functional
\begin{equation}
\label{Del0}
    \Delta(t)
     :=
    \bE
    Q(t)
    =
    \bra
    \Sigma,
    \Re \Ups(t)
    \ket=
    \|F \alpha_t \sqrt{P}\|^2 + \bra \Sigma, \Re V(t) \ket.
\end{equation}
This quantity provides an  upper bound  on the tail probabilities for the positive semi-definite self-adjoint quantum variable $Q(t)$ in (\ref{Q}), (\ref{Qeta}):
\begin{equation}
    \label{tail}
        \bP_t([z,+\infty))
        \<
        \frac{\Delta(t)}{z},
        \qquad
        z > \Delta(t),
\end{equation}
where Markov's inequality \cite{S_1996}
is applied to the probability distribution $\bP_t(\cdot)$ of $Q(t)$
on $[0,+\infty)$ obtained by averaging its spectral measure \cite{H_2001}  over the quantum state (\ref{rho}). In (\ref{Del0}), use is made of the Frobenius norm $\|\cdot\|$   and inner product  $\bra\cdot, \cdot\ket $   of matrices \cite{HJ_2007} along with the one-point second-moment matrix  of the process $\xi$:
$    \Ups(t)
     :=
    \bE (\xi(t)\xi(t)^\rT)
    =
    \alpha_t
    \Pi
    \alpha_t^\rT
    +
    V(t)
$.
The latter is computed in \cite{VP_2024_ANZCC} by using (\ref{XZcomm}), (\ref{EXZ}), (\ref{xi}), (\ref{exp}) and the second-moment matrix
\begin{equation}
\label{EXX0}
    \Pi
     :=
    \bE (X(0)X(0)^\rT)
    =
    P + i\Theta,
    \qquad
    P:= \Re \Pi
\end{equation}
of the initial system variables $X_1(0), \ldots, X_n(0)$, including  their CCR matrix $\Theta$ from (\ref{Xqpcomm}),
and the quantum covariance matrix
\begin{equation}
\label{V}
    V(t)
      :=
    \bE(Z(t)Z(t)^\rT)
     =
    \int_0^t
    \re^{sA}
    B
    \Omega
    B^\rT
    \re^{sA^\rT}
    \rd s,
    \qquad
    t \> 0
\end{equation}
of the process $Z$ in (\ref{Xsol}), with
$\Omega$ the quantum Ito matrix from (\ref{Omega}). The matrix $V(t)$ in (\ref{V}),  which  coincides with the controllability Gramian of the pair $(A, B\sqrt{\Omega})$ over the time interval $[0,t]$,  satisfies  the initial value problem for the Lyapunov ODE
\begin{equation}
\label{VALE}
    \mathop{V}^{\centerdot}(t)
    =
    AV(t) + V(t)A^\rT + \mho,
    \qquad
    t \> 0,
    \quad
    V(0) = 0,
\end{equation}
where
\begin{equation}
\label{BOB}
    \mho := B\Omega B^\rT.
\end{equation}
From (\ref{VALE}), the time derivatives $V^{(k)}:= \rd^k V/\rd t^k$ can be computed recursively as
$
    V^{(k+1)} = AV^{(k)} + V^{(k)}A^\rT  + \delta_{k0} \mho$ for any $
    k \> 0
$,
where $\delta_{jk}$ is the Kronecker delta. In particular, the first three derivatives at $t=0$  are
\begin{align}
\label{V1}
    \dot{V}(0)
    & = \mho,\\
\label{V2}
    \ddot{V}(0)
    & = A\mho  + \mho A^\rT ,\\
\label{V3}
    \dddot{V}(0)
    & =
    A^2\mho + \mho (A^\rT)^2 + 2 A \mho A^\rT.
\end{align}
In what follows, we will use the short-horizon asymptotic behaviour of (\ref{Del0}) described below.

\begin{lem}
\label{lem:Del0asy}
Suppose the condition (\ref{Mrank}) of Lemma~\ref{lem:iso} is satisfied and the matrix $F$ in (\ref{FF}) is chosen  so as to satisfy (\ref{FB0}). Then the mean-square deviation functional (\ref{Del0}) behaves asymptotically as
\begin{equation}
\label{Del0asy}
  \Delta(t)
  =
  \|G \sqrt{P}\|^2 t^2 + O(t^3),
  \qquad
  {\rm as}\
  t \to 0+,
\end{equation}
where the coefficient is
computed in terms of the matrices $G$ from (\ref{phidot}) and $P$ from (\ref{EXX0}).
\hfill$\square$
\end{lem}
\begin{proof}
From a combination of (\ref{Del0})  with (\ref{VALE})--(\ref{V2}), it follows that
\begin{align}
\label{D0}
    \Delta(0)
    & = \bra \Sigma, \Re V(0)\ket  =0,\\
\label{D1}
    \dot{\Delta}(0)
    & = \bra \Sigma, \Re \dot{V}(0)\ket  =\|FB\|^2,\\
\nonumber
    \ddot{\Delta}(0)
    & =
    2 \|F A \sqrt{P}\|^2 +
    \bra \Sigma, \Re \ddot{V}(0)\ket\\
\label{D2}
    & =
    2 (\|F A \sqrt{P}\|^2 +
     \bra FB, FAB\ket) .
\end{align}
Now, let the matrix $F$ satisfy (\ref{FB0}). The existence of such an $F$ is guaranteed by (\ref{Mrank}) according to Lemma~\ref{lem:iso} and its proof. Then $F\dot{V}(0) F^\rT
     =
    F\mho F^\rT = 0$, $F \ddot{V}(0) F^\rT
     =
    F(A\mho  + \mho A^\rT)F^\rT  =0$ and $
    F \dddot{V}(0)F^\rT
     =
    F(A^2\mho + \mho (A^\rT)^2 + 2 A \mho A^\rT)F^\rT
    =
    2G\mho G^\rT$
in view of (\ref{BOB})--(\ref{V3}),
so that
the external fields manifest themselves in (\ref{Fxi}) starting from the third-order term of the Taylor series approximation
\begin{equation}
\label{covFZ}
    \cov(FZ(t))
    =
    FV(t)F^\rT
    =
    \frac{t^3}{3} G\mho G^\rT + O(t^4),
    \qquad
    {\rm as}\
    t \to 0+.
\end{equation}
Accordingly, their contributions through the terms $\bra \Sigma, \Re V^{(k)}(0)\ket$,  with $k=1, 2$,  vanish, thus reducing
(\ref{D1}), (\ref{D2}) to
\begin{equation}
\label{Del0ders}
    \dot{\Delta}(0) = 0,
    \qquad
  \ddot{\Delta}(0)
  =
  2\|G \sqrt{P}\|^2,
\end{equation}
where use is also made of (\ref{FAG}). Substitution of (\ref{D0}), (\ref{Del0ders}) into an appropriately   truncated Taylor series expansion of $\Delta$   leads to (\ref{Del0asy}).
\end{proof}

Due to the special choice of the matrix $F$ in Lemmas~\ref{lem:iso} and \ref{lem:Del0asy},  the leading term in (\ref{Del0asy}) is quadratic (rather than linear) in time. This is a consequence of the partial isolation of the subsystem $\Phi$ from the external fields (see (\ref{covFZ}) and Fig.~\ref{fig:PhiPsi}) and  makes the mean-square deviation functional $\Delta$ grow slower (at least at short time horizons) compared to the case of arbitrary matrices $F$ considered in \cite{VP_2024_ANZCC}.

From the viewpoint of optimizing a network of OQHOs as a quantum memory system,  the minimization  (at a suitably chosen time $t> 0$) of the mean-square deviation functional $\Delta(t)$ in (\ref{Del0}) over admissible energy and coupling parameters of the network also minimizes the tail probability bounds (\ref{tail}) and  provides a relevant performance criterion  for such applications. This minimization improves
the ability of the OQHO   as a quantum memory  to retain its initial system variables and   is closely related (by duality) to maximizing the  decoherence time
\begin{equation}
\label{tau0}
    \tau(\eps)
    :=
    \inf
    \{
        t\> 0:\
        \Delta(t)
        >
        \eps \|F\sqrt{P}\| ^2
    \}
\end{equation}
proposed in \cite{VP_2024_ANZCC,VP_2024_EJC} as a horizon by which the selected system variables do not deviate too far from their initial values.
In accordance with (\ref{FF}), (\ref{EXX0}),
the quantity $\bE ((FX(0))^\rT FX(0) ) = \bra \Sigma, \Pi\ket = \|F\sqrt{P}\|^2$ in (\ref{tau0}) provides a reference scale, with respect to which the  dimensionless parameter
$\eps>0$ specifies a relative error threshold  for $F X(t)$ to approximately reproduce $F X(0)$. Since the columns of the matrix $F$ in (\ref{FF}) are linearly dependent if $s<n$, we assume that
\begin{equation}
\label{Fpos}
    F
    \sqrt{P}
    \ne 0
\end{equation}
in order to avoid the trivial case $\tau(\eps) = 0$.
In the partial isolation setting of Lemmas~\ref{lem:iso} and~\ref{lem:Del0asy}, where $F$ satisfies (\ref{FB0}) (as opposed to the case of $FB\ne 0$ studied in \cite{VP_2024_ANZCC,VP_2024_EJC}), the asymptotic behaviour of (\ref{tau0}) in the high-fidelity limit (of small values of $\eps$) is as follows.

\begin{thm}
\label{th:tauasy}
Suppose the conditions of Lemma~\ref{lem:Del0asy} are satisfied together with (\ref{Fpos}) and
\begin{equation}
\label{Gpos}
    G
    \sqrt{P}
    \ne 0.
\end{equation}
Then the memory decoherence time (\ref{tau0}) behaves asymptotically as
\begin{equation}
\label{tau0asy}
    \tau(\eps)
    \sim
    \frac{\|F\sqrt{P}\|}{\|G\sqrt{P}\|}
    \sqrt{\eps}
    =:
    \wh{\tau}(\eps),
    \qquad
    {\rm as}\
    \eps\to 0+.
\end{equation}
\hfill$\square$
\end{thm}
\begin{proof}
By (\ref{tau0}), the memory decoherence time $\tau(\eps)$ is a nondecreasing function of the fidelity parameter $\eps$ satisfying
\begin{equation}
\label{level}
    \Delta(\tau(\eps))
    =
    \eps \|F \sqrt{P}\|^2,
\end{equation}
so that $\tau(\eps)>0$ for any $\eps>0$ due to (\ref{Fpos}) and since $\Delta(t)$ in (\ref{Del0}) is a continuous function of $t\> 0$,   with $\Delta(0) = 0$. Furthermore,
\begin{equation}
\label{taulim}
    \lim_{\eps \to 0+} \tau(\eps) = 0.
\end{equation}
By combining (\ref{level}), (\ref{taulim}) with the asymptotic relation (\ref{Del0asy}) under the condition (\ref{Gpos}), it follows that
$
    \eps\|F\sqrt{P}\|^2
    \sim
    \|G \sqrt{P}\|^2 \tau(\eps)^2
$,     as $
    \eps \to 0+$,
which leads to (\ref{tau0asy}).
\end{proof}

As a consequence of the partial subsystem isolation,   the asymptotic behaviour (\ref{tau0asy}) is qualitatively different and yields a longer decoherence time compared to the case of $FB\ne 0$ in \cite{VP_2024_ANZCC}, where $\tau(\eps)$ is asymptotically linear with respect to $\eps$. In the framework of (\ref{tau0asy}),  the maximization of $\tau(\eps)$ at a given small value of $\eps$ can be replaced with its ``approximate''  version
\begin{equation}
\label{tauhatmax}
  \wh{\tau}(\eps)
  \longrightarrow
  \sup.
\end{equation}
With $F$ and $P$ being fixed, (\ref{tauhatmax}) is equivalent to
the minimization of the denominator
\begin{equation}
\label{GP}
    \|G \sqrt{P}\| = 2\|F\Theta R \sqrt{P}\|
\end{equation}
and is organised as a convex quadratic optimization problem over the energy matrix $R$.

\section{MEMORY DECOHERENCE TIME MAXIMIZATION FOR OQHO INTERCONNECTION}
\label{sec:two}

Consider  the approximate memory decoherence time maximization (\ref{tauhatmax})  for a coherent feedback interconnection of two OQHOs from \cite[Section~8]{VP_2023_SCL} and \cite{VP_2024_ANZCC},  which interact with external input bosonic fields and    are coupled to each other through a direct energy coupling and an indirect field-mediated coupling \cite{ZJ_2011a}; see
Fig.~\ref{fig:system}.
\begin{figure}[htbp]
\centering
\unitlength=0.85mm
\linethickness{0.4pt}
\begin{picture}(50.00,15.00)
    \put(10,1){\framebox(10,10)[cc]{\scriptsize OQHO${}_1$}}
    \put(30,1){\framebox(10,10)[cc]{\scriptsize OQHO${}_2$}}
    \put(0,6){\vector(1,0){10}}
    \put(50,6){\vector(-1,0){10}}
    \put(20,6){\vector(1,0){10}}
    \put(30,6){\vector(-1,0){10}}
    \put(20,2){\vector(1,0){10}}
    \put(30,10){\vector(-1,0){10}}
    \put(-2,6){\makebox(0,0)[rc]{\small $W^{(1)}$}}
    \put(25,14){\makebox(0,0)[cc]{\small $Y^{(2)}$}}
    \put(52,6){\makebox(0,0)[lc]{\small $W^{(2)}$}}
    \put(25,-1){\makebox(0,0)[cc]{\small $Y^{(1)}$}}
\end{picture}
\caption{
    A coherent feedback interconnection of two OQHOs, interacting with external input quantum Wiener processes  $W^{(1)}$, $W^{(2)}$  and coupled to each other through a direct energy coupling (represented by a double-headed arrow) and a field-mediated coupling through the quantum Ito processes $Y^{(1)}$, $Y^{(2)}$ at the corresponding outputs.
}
\label{fig:system}
\end{figure}
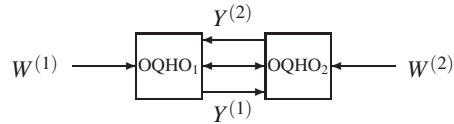
In this  setting, the
quantum Wiener processes  $W^{(1)}$, $W^{(2)}$ of the external  fields  of even dimensions $m_1$, $m_2$ are in the vacuum states $\ups_1$, $\ups_2$ on symmetric Fock spaces $\fF_1$, $\fF_2$, respectively, so that $W^{(1)}$, $W^{(2)}$   commute with and are statistically independent of each other.  The augmented quantum Wiener process
$    W
    :=
    {\small\begin{bmatrix}
      W^{(1)}\\
      W^{(2)}
    \end{bmatrix}}
$ of dimension $m:= m_1+m_2$
on the composite Fock space $\fF:= \fF_1\ox \fF_2$ has the quantum Ito matrix $\Omega
    =
    {\small\begin{bmatrix}
      \Omega_1 & 0\\
      0 & \Omega_2
    \end{bmatrix}}$ in (\ref{Omega}) coming from the individual Ito tables $
    \rd W^{(k)}\rd W^{(k)}{}^{\rT} = \Omega_k \rd t$, where  $
    J=
    {\small\begin{bmatrix}
      J_1 & 0\\
      0 & J_2
    \end{bmatrix}}$,
$
    \Omega_k
     := I_{m_k} + iJ_k$ and
$
    J_k:= I_{m_k/2} \ox \bJ
$,
with the matrix
$\bJ$ from (\ref{JJ}). The component OQHOs are endowed with initial spaces $\fH_k$ and vectors $X^{(k)}$ of even numbers $n_k:= 2\nu_k$ of dynamic variables on the composite system-field space (\ref{fH}), where $\fH_0:= \fH_1 \ox \fH_2$. Accordingly, the vector $    X
    :=
    {\small\begin{bmatrix}
      X^{(1)}\\
      X^{(2)}
    \end{bmatrix}}
$ of $n:= n_1+n_2$ system variables with $\nu = \nu_1 + \nu_2$ in (\ref{nnu})  for the augmented OQHO   satisfies (\ref{Xqpcomm}) with the CCR matrix
\begin{equation}
\label{TTT}
     \Theta
    :=
    {\begin{bmatrix}
      \Theta_1 & 0\\
      0 & \Theta_2
    \end{bmatrix}}
\end{equation}
which consists of the individual CCR matrices $\Theta_k = \frac{1}{2}I_{\nu_k} \ox \bJ$:
\begin{equation}
\label{Theta12}
    [X^{(1)},
      X^{(2)}{}^\rT] = 0,
      \qquad
    [X^{(k)},
      X^{(k)}{}^\rT] = 2i\Theta_k,
      \qquad
      k = 1, 2.
\end{equation}
In addition to the individual Hamiltonians $H_k:= \frac{1}{2} X^{(k)}{}^\rT R_k X^{(k)}$ of the OQHOs, parameterized by their energy matrices $R_k = R_k^\rT \in \mR^{n_k \x n_k}$, the direct energy  coupling (see Fig.~\ref{fig:system}) contributes the term
\begin{equation}
\label{Hint}
    H_{12}
    :=
    X^{(1)}{}^\rT R_{12} X^{(2)} = X^{(2)}{}^\rT R_{21} X^{(1)}
\end{equation}
specified by
\begin{equation}
\label{R12}
    R_{12} = R_{21}^\rT \in \mR^{n_1\x n_2}
\end{equation}
in view of the commutativity in (\ref{Theta12}). Since there is also an additional indirect coupling of the OQHOs, which is mediated by their output fields $Y^{(1)}$, $Y^{(2)}$ of even dimensions $r_1$, $r_2$, the resulting interconnection is
governed by the QSDEs \cite{VP_2024_ANZCC}
\begin{align}
\label{x}
    \rd X^{(k)}
    & =
    (A_k X^{(k)} + F_k X^{(3-k)}) \rd t  +  B_k \rd W^{(k)}  + E_k \rd Y^{(3-k)} ,\\
\label{y}
    \rd Y^{(k)}
    & =
    C_k X^{(k)} \rd t  +  D_k \rd W^{(k)} ,
    \qquad
    k = 1, 2.
\end{align}
The matrices
$    A_k\in \mR^{n_k\x n_k}$,
    $
    B_k\in \mR^{n_k\x m_k}$,
    $
    C_k\in \mR^{r_k\x n_k}$,
    $
    E_k\in \mR^{n_k\x r_{3-k}}$,
    $
    F_k\in \mR^{n_k\x n_{3-k}}
$
are parameterized as
\begin{align}
\label{Ak_Bk}
    A_k
     & =
    2\Theta_k(R_k + M_k^{\rT}J_k M_k + N_k^{\rT}\wt{J}_{3-k}N_k),
    \quad
        B_k
     = 2\Theta_k M_k^{\rT},\\
\label{Ck_Ek_Fk}
    C_k  & =2D_kJ_k M_k,
    \qquad
    E_k  = 2\Theta_k N_k^{\rT},
    \qquad
    F_k  = 2\Theta_k R_{k,3-k},
\end{align}
and the matrices     $
    D_k\in \mR^{r_k\x m_k}$ consist of  $r_k \< m_k$ conjugate rows of permutation matrices of orders $m_k$ with $D_k \Omega_k D_k^\rT = I_{r_k} + i\wt{J}_k$, where
$    \wt{J}_k:= D_kJ_kD_k^{\rT}  = I_{r_k/2}\ox \bJ$.
Here,
$M_k\in \mR^{m_k \x n_k}$,  $N_k\in \mR^{r_{3-k} \x n_k}$ are the matrices of coupling of the $k$th  OQHO to its external input field $W^{(k)}$ and the output $Y^{(3-k)}$ of the other OQHO, respectively.  By (\ref{x}), (\ref{y}),  the composite OQHO in Fig.~\ref{fig:system}  is governed by the first QSDE in (\ref{dXdY_OQHO}) with the matrices
\begin{equation}
\label{cAB}
    A
    =
    {\begin{bmatrix}
        A_1 & F_1+ E_1C_2\\
        F_2 +E_2C_1 & A_2
    \end{bmatrix}},
    \qquad
    B
    =
    {\begin{bmatrix}
        B_1 & E_1D_2\\
        E_2D_1 & B_2
    \end{bmatrix}}
\end{equation}
and,
by (\ref{Ak_Bk}), (\ref{Ck_Ek_Fk}),  has the following energy and coupling matrices:
\begin{equation}
\label{Rclos_Nclos}
    R  = R_* +   {\begin{bmatrix}
    0 & R_{12}\\
    R_{12}^\rT  & 0
  \end{bmatrix}},
  \qquad
    M   =
    {\begin{bmatrix}
      M_1 & D_1^{\rT}N_2 \\
      D_2^{\rT}N_1 & M_2
    \end{bmatrix}}.
\end{equation}
Here, the matrix
\begin{equation}
\label{R*}
    R_*
    :=
    {\begin{bmatrix}
      R_1                                       & N_1^{\rT}D_2J_2 M_2 -M_1^\rT J_1 D_1^\rT N_2\\
      N_2^{\rT}D_1J_1 M_1 -M_2^\rT J_2 D_2^\rT N_1  & R_2
    \end{bmatrix}}
\end{equation}
is associated with the field-mediated coupling between the constituent OQHOs, their coupling to the external fields,  and the individual Hamiltonians. Therefore, if the matrices $M_k$, $N_k$, $R_k$  are fixed for all $k=1,2$, and hence, so also is $R_*$ in (\ref{R*}), then the energy matrix $R$ in (\ref{Rclos_Nclos}) can only be varied over a proper affine subspace in the subspace of real symmetric matrices of order $n$ by varying the direct energy coupling matrix $R_{12}$. This leads to  the following formulation of the approximate memory decoherence time maximization problem
(\ref{tauhatmax}):
\begin{equation}
\label{tauhatmaxR12}
  \wh{\tau}(\eps) \longrightarrow \sup
  \qquad{\rm over}\
  R_{12}\in \mR^{n_1\x n_2}.
\end{equation}
This setting is similar to that in \cite{VP_2024_ANZCC}, except that we are concerned here with  a subsystem of the closed-loop OQHO which is partially isolated from the external fields $W^{(1)}$, $W^{(2)}$.

\begin{thm}
\label{th:R12}
Suppose the OQHO interconnection,  described by (\ref{TTT})--(\ref{R*}),  and the matrix $F$ in (\ref{FF}) satisfy the conditions of Theorem~\ref{th:tauasy}.  Then the direct energy coupling matrix $R_{12}$ in (\ref{Hint}), (\ref{R12}) solves the problem (\ref{tauhatmaxR12}) with the approximate memory decoherence time $\wh{\tau}$ in
(\ref{tau0asy}) for the composite OQHO if and only if
\begin{equation}
\label{gK0}
    g(R_{12}) + K = 0,
\end{equation}
where $g$ is a linear operator acting on a matrix $N \in \mR^{n_1 \x n_2}$  as
\begin{align}
\nonumber
    g(N)
     := &
    \Theta_1 \Sigma_{11}\Theta_1 N P_{22}   +
    P_{11}N \Theta_2 \Sigma_{22}\Theta_2    \\
\label{g}
     & +   \Theta_1 \Sigma_{12}\Theta_2 N^\rT  P_{12}
    +
    P_{12}N^\rT \Theta_1 \Sigma_{12}\Theta_2    ,
\end{align}
and $K \in \mR^{n_1\x n_2}$ is an auxiliary matrix  which is computed as
\begin{equation}
\label{K}
    K:=
    2
    (\bS
    (\Theta \Sigma \Theta R_*P))_{12}
\end{equation}
in terms of (\ref{FF}),  (\ref{EXX0}),  (\ref{R*}),
with $(\cdot)_{jk}$ the appropriate matrix blocks, and $\bS(\gamma):= \frac{1}{2}(\gamma+\gamma^\rT)$ the matrix symmetrizer. \hfill$\square$
\end{thm}
\begin{proof}
As mentioned above, (\ref{tauhatmaxR12}) reduces to the minimization of the quantity (\ref{GP}), which is equivalent to minimizing the convex quadratic function
\begin{equation}
\label{fR12}
    f(R_{12})
    :=
    \frac{1}{2}\|F\Theta R \sqrt{P}\|^2
\end{equation}
over $R_{12} \in \mR^{n_1\x n_2}$, where the $\frac{1}{2}$-factor  is introduced for convenience, and $R$ is an affine function of $R_{12}$ in (\ref{Rclos_Nclos}).   Therefore, $R_{12}$ delivers a global minimum to (\ref{fR12}) if and only if it makes the Frechet derivative \cite{RS_1980} of $f$ (on the Hilbert space $\mR^{n_1\x n_2}$ with the Frobenius inner product $\bra\cdot, \cdot\ket $)  vanish:
\begin{equation}
\label{f'0}
    f'(R_{12}) = 0.
\end{equation}
In view of the variational identity $\delta (\|\phi\|^2) = 2\bra \phi, \delta \phi\ket $, the first variation of (\ref{fR12}) with respect to the matrix $R_{12}$ takes the form
\begin{align}
\nonumber
    \delta f(R_{12})
    & =
    \bra
    F\Theta R \sqrt{P},
    F\Theta (\delta R) \sqrt{P}
    \ket
    =
    -
    \bra
    \Theta F^\rT F\Theta R P,
    \delta R
    \ket\\
\nonumber
    & =
    -
    \bra
    (\Theta \Sigma\Theta R P)_{12},
    \delta R_{12}
    \ket
    -
    \bra
    (\Theta \Sigma\Theta R P)_{21},
    \delta R_{12}^\rT
    \ket\\
\label{deltaf}
    & =
    -
    \bra
    (\Theta \Sigma\Theta R P)_{12}
    +
    (\Theta \Sigma\Theta R P)_{21}^\rT,
    \delta R_{12}
    \ket,
\end{align}
where use is also made of  the antisymmetry of $\Theta$ in (\ref{TTT}) and symmetry of $P$ in (\ref{EXX0}) along with (\ref{FF}) and $\delta R  = {\small\begin{bmatrix}
  0 & \delta R_{12}\\
  \delta R_{12}^\rT & 0
\end{bmatrix}}$ since the matrix $R_*$ in (\ref{Rclos_Nclos}) is fixed. From (\ref{deltaf}), it follows  that
\begin{equation}
\label{f'}
    -f'(R_{12})
    =
    2(\bS(\Theta \Sigma\Theta R P))_{12}
    =
    g(R_{12}) + K,
\end{equation}
with $g$, $K$ from (\ref{g}), (\ref{K}). By (\ref{f'}), (\ref{f'0}) is equivalent to (\ref{gK0}), thus establishing
the latter as a necessary and sufficient condition of optimality for (\ref{tauhatmaxR12}).
\end{proof}

In view of (\ref{f'}), the linear map $g$ in (\ref{g}) is a negative semi-definite self-adjoint operator on $\mR^{n_1\x n_2}$  related to the second Frechet derivative $f''$  of the function $f$ in (\ref{fR12}) by $g = -f''$. Also note that the linear equation (\ref{gK0}) can be solved by vectorizing \cite{M_1988} the matrix $R_{12}$.


\begin{thebibliography}{99}{
\bibitem{A_1989}
V.I.Arnold, \emph{Mathematical Methods of Classical Mechanics}, 2nd Ed., Springer, New York, 1989.


\bibitem{DJ_2006}
C.D'Helon, and M.R.James, Stability, gain, and robustness in quantum feedback networks,  \textit{Phys. Rev.  A}, vol. 73, no. 5,  2006, 053803.

\bibitem{EMPUJ_2016}
L.A.D.Espinosa, Z.Miao, I.R.Petersen, V.Ugrinovskii, and
M.R.James,
Physical realizability and preservation of
commutation and anticommutation relations for
$n$-level quantum systems,
\textit{SIAM J. Control Optim.},
vol. 54, no. 2, 2016, pp. 632--661.



\bibitem{FCHJ_2016}
S.Fu, A.R.R.Carvalho, M.R.Hush, and M.R.James,
Cross-phase modulation and entanglement in a compound
gradient echo memory,
\emph{Phys. Rev. A},  vol. 93, 2016, 023809.


\bibitem{GKS_1976}
V.Gorini, A.Kossakowski, E.C.G.Sudarshan, Completely positive dynamical semigroups of $N$-level systems,
\emph{J. Math. Phys.}, vol. 17, no. 5, 1976, pp. 821--825.


\bibitem{HEHBANS_2016}
K.Heshami, D.G.England, P.C.Humphreys, P.J.Bustard, V.M.Acosta, J.Nunn, B.J.Sussman,  Quantum memories: emerging applications and recent advances, \emph{ J. Mod. Opt.},  vol. 63, no. 20, 2016,  pp. 2005--2028.



\bibitem{H_2001}
A.S.Holevo, \textit{Statistical Structure of Quantum Theory}, Springer, Berlin, 2001.

\bibitem{H_2019}
A.S.Holevo, \textit{Quantum Systems, Channels, Information: A Mathematical Introduction}, De Gruyter, Berlin, 2019.

\bibitem{HJ_2007}
R.A.Horn, and C.R.Johnson,
\textit{Matrix Analysis},
Cambridge
University Press, New York, 2007.

\bibitem{HP_1984}
R.L.Hudson,  and K.R.Parthasarathy,
Quantum Ito's formula and stochastic evolutions,
\textit{Commun. Math. Phys.}, vol.  93, 1984, pp. 301--323.

\bibitem{JNP_2008}
M.R.James, H.I.Nurdin, and I.R.Petersen,
$H^{\infty}$ control of
linear quantum stochastic systems,
\textit{IEEE Trans.
Automat. Contr.}, vol. 53, no. 8, 2008, pp. 1787--1803.


\bibitem{KS_1991}
I.Karatzas, and S.E.Shreve,
\emph{Brownian Motion and Stochastic Calculus}, 2nd Ed.,
Springer-Verlag,  New York, 1991.


\bibitem{LL_1981}
L.D.Landau, and E.M.Lifshitz, \emph{Quantum Mechanics: Non-relativistic Theory}, 3rd Ed., Butterworth-Heinemann, Amsterdam, 1981.

\bibitem{L_1976}
G.Lindblad, On the generators of quantum dynamical semigroups,
\emph{Comm. Math. Phys.}, vol. 48, 1976, pp. 119--130.



\bibitem{M_1988}
J.R.Magnus, \textit{Linear Structures}, Oxford University Press, New York, 1988.

\bibitem{NC_2000}
M.A.Nielsen, and I.L.Chuang,
\textit{Quantum Computation and Quantum Information},
Cambridge University Press, Cambridge, 2000.


\bibitem{NY_2017}
H.I.Nurdin, and N.Yamamoto,
\textit{Linear Dynamical Quantum Systems},
Springer, Netherlands, 2017.


\bibitem{P_1992}
K.R.Parthasarathy,
\textit{An Introduction to Quantum Stochastic Calculus},
Birk\-h\"{a}user, Basel, 1992.


\bibitem{P_2017}
I.R.Petersen,
Quantum linear systems theory,
\textit{Open Automat. Contr. Syst. J.},
vol. 8, 2017, pp. 67--93.

\bibitem{RS_1980}
M.Reed, and B.Simon, \emph{Methods of Modern Mathematical Physics: Functional Analysis}, vol. 1, 2nd Ed., Academic Press, San Diego, 1980.


\bibitem{S_1994}
J.J.Sakurai,
\textit{Modern Quantum Mechanics},
 Addison-Wesley, Reading, Mass., 1994.

\bibitem{S_1996}
A.N.Shiryaev, \emph{Probability}, 2nd Ed., Springer, New York, 1996.



\bibitem{VPJ_2018a}
I.G.Vladimirov, I.R.Petersen, and M.R.James, Multi-point Gaussian states, quadratic-exponential cost functionals, and large deviations estimates for linear quantum stochastic systems, \textit{Appl. Math. Optim.}, vol. 83, 2021, pp. 83--137 (published 24 July 2018).


\bibitem{VP_2022_SIAM}
I.G.Vladimirov, and I.R.Petersen,
Moment dynamics and observer design for a class of quasilinear quantum stochastic systems,
\emph{SIAM J. Control Optim.}, vol. 60, no. 3, 2022, pp. 1223--1249.


\bibitem{VP_2023_SCL}
I.G.Vladimirov, and I.R.Petersen,
Decoherence quantification through commutation relations decay for open quantum harmonic oscillators,
\emph{Systems \& Control Letters},
vol. 178,
2023,
105585.


\bibitem{VP_2024_ANZCC}
I.G.Vladimirov,  and I.R.Petersen,
Decoherence time in quantum harmonic oscillators as quantum memory systems, 2024 Australian and New Zealand Control Conference (ANZCC), Gold Coast, Australia, 2024, pp. 211--216.



\bibitem{VP_2024_EJC}
I.G.Vladimirov,  and I.R.Petersen, Decoherence time control by interconnection for finite-level quantum memory systems, \emph{European Journal of Control}, 2024, 101054,
published online 15 June 2024, \url{https://doi.org/10.1016/j.ejcon.2024.101054}.


\bibitem{WM_2008}
D.F.Walls, and G.J.Milburn,
\emph{Quantum Optics}, Springer, Berlin, 2008.

\bibitem{WLZ_2019}
Y.Wang, J.Li, S.Zhang, et al,  Efficient quantum memory for single-photon polarization qubits, \emph{ Nat. Photonics}, vol. 13, 2019, pp. 346--351.


\bibitem{YJ_2014}
N.Yamamoto,  and M.R.James,
Zero-dynamics principle for perfect quantum
memory in linear networks,
\emph{New J. Phys.}, vol.  16,  2014, 073032.


\bibitem{ZD_2022}
G.Zhang, and Z.Dong,
Linear quantum systems: a tutorial,
\emph{Annual Reviews in Control},
vol. 54,
2022,
pp. 274--294.

\bibitem{ZJ_2011a}
G.Zhang, and M.R.James,  Direct and indirect couplings in coherent feedback control of linear quantum systems,
\textit{IEEE Trans. Automat. Contr.}, vol. 56, no. 7, 2011, 1535--1550.

}
\end{thebibliography}
\end{document}